\documentclass[12pt]{article}%

\IfFileExists{sariel_computer.sty}{\def\sarielComp{1}}{}
\ifx\sarielComp\undefined%
\newcommand{\SarielComp}[1]{}
\newcommand{\NotSarielComp}[1]{#1}%
\else
\newcommand{\SarielComp}[1]{#1}%
\newcommand{\NotSarielComp}[1]{}%
\fi
\newcommand{\IfPrinterVer}[2]{#2}%

\usepackage[cm]{fullpage}%
\usepackage{amsmath}%
\usepackage{amssymb}%
\usepackage{xcolor}%
\usepackage{graphicx}%
\usepackage{caption}%
\usepackage{stmaryrd}%

\definecolor{blue25}{rgb}{0,0,0}
\providecommand{\emphic}[2]{%
   \textcolor{blue25}{%
      \textbf{\emph{#1}}}%
   \index{#2}}

\providecommand{\emphi}[1]{}%
\renewcommand{\emphi}[1]{\emphic{#1}{#1}}

\SarielComp{\usepackage{sariel_colors}}%

\usepackage[amsmath,thmmarks]{ntheorem}%
\theoremseparator{.}%

\usepackage{titlesec}%
\titlelabel{\thetitle. }%
\usepackage{xcolor}%
\usepackage{mleftright}%
\usepackage{xspace}%
\usepackage[inline]{enumitem}%
\usepackage{hyperref}%

\IfPrinterVer{%
   \usepackage{hyperref}%
}{%
   \usepackage{hyperref}%
   \hypersetup{%
      breaklinks,%
      ocgcolorlinks, colorlinks=true,%
      urlcolor=[rgb]{0.25,0.0,0.0},%
      linkcolor=[rgb]{0.5,0.0,0.0},%
      citecolor=[rgb]{0,0.2,0.445},%
      filecolor=[rgb]{0,0,0.4},
      anchorcolor=[rgb]={0.0,0.1,0.2}%
   }
}

\newcommand{\hrefb}[3][black]{\href{#2}{\color{#1}{#3}}}%

\theoremseparator{.}%

\newlist{compactenumA}{enumerate}{5}%
\setlist[compactenumA]{topsep=0pt,itemsep=-1ex,partopsep=1ex,parsep=1ex,%
   label=(\Alph*)}%

\theoremstyle{plain}%
\newtheorem{theorem}{Theorem}[section]

\newtheorem{lemma}[theorem]{Lemma}

\newtheorem{observation}[theorem]{Observation}

\theoremstyle{plain}%
\theoremheaderfont{\sf} \theorembodyfont{\upshape}%
\newtheorem*{remark:unnumbered}[theorem]{Remark}%

\newtheorem{defn}[theorem]{Definition}

\newcommand{\myqedsymbol}{\rule{2mm}{2mm}}

\theoremheaderfont{\em}%
\theorembodyfont{\upshape}%
\theoremstyle{nonumberplain}%
\theoremseparator{}%
\theoremsymbol{\myqedsymbol}%
\newtheorem{proof}{Proof:}%

\newcommand{\atgen}{\symbol{'100}}
\newcommand{\SarielThanks}[1]{\thanks{Department of Computer Science;
      University of Illinois; 201 N. Goodwin Avenue; Urbana, IL,
      61801, USA; {\tt sariel\atgen{}illinois.edu}; {\tt
         \url{http://sarielhp.org/}.} #1}}

\newcommand{\AmirThanks}[1]{%
   \thanks{%
      School of Electrical Engineering and Computer Science; %
      Oregon State University;
      \href{mailto:nayyeria@eecs.oregonstate.edu}
      {{nayyeria@eecs.oregonstate.edu}}. #1
   }%
}

\newcommand{\HLink}[2]{\hyperref[#2]{#1~\ref*{#2}}}
\newcommand{\HLinkSuffix}[3]{\hyperref[#2]{#1\ref*{#2}{#3}}}

\newcommand{\figlab}[1]{\label{fig:#1}}
\newcommand{\figref}[1]{\HLink{Figure}{fig:#1}}

\newcommand{\thmlab}[1]{{\label{theo:#1}}}
\newcommand{\thmref}[1]{\HLink{Theorem}{theo:#1}}

\newcommand{\apndlab}[1]{\label{apnd:#1}}

\newcommand{\itemlab}[1]{\label{item:#1}}
\newcommand{\itemref}[1]{\HLinkSuffix{}{item:#1}{}}

\newcommand{\lemlab}[1]{\label{lemma:#1}}
\newcommand{\lemref}[1]{\HLink{Lemma}{lemma:#1}}%

\newcommand{\seclab}[1]{\label{sec:#1}}
\newcommand{\secref}[1]{\HLink{Section}{sec:#1}}

\providecommand{\eqlab}[1]{}%
\renewcommand{\eqlab}[1]{\label{equation:#1}}

\newcommand{\remove}[1]{}%

\newcommand{\pth}[2][\!]{\mleft({#2}\mright)}%

\newcommand{\ceil}[1]{\left\lceil {#1} \right\rceil}
\newcommand{\floor}[1]{\left\lfloor {#1} \right\rfloor}

\newcommand{\brc}[1]{\left\{ {#1} \right\}}
\newcommand{\cardin}[1]{\left| {#1} \right|}%

\renewcommand{\th}{th\xspace}
\providecommand{\Mh}[1]{#1}%

\newcommand{\lenX}[1]{\cardin{#1}}%

\newcommand{\etal}{\textit{et~al.}\xspace}

\renewcommand{\Re}{\mathbb{R}}%

\newcommand{\Term}[1]{\textsf{#1}}

\newcommand{\BFS}{\Term{BFS}\index{BFS}\xspace}
\newcommand{\DFS}{\Term{BFS}\index{DFS}\xspace}

\numberwithin{figure}{section}%
\numberwithin{table}{section}%
\numberwithin{equation}{section}%

\newlength{\savedparindent}
\newcommand{\SaveIndent}{\setlength{\savedparindent}{\parindent}}
\newcommand{\RestoreIndent}{\setlength{\parindent}{\savedparindent}}

\newlength{\ppicX}
\newlength{\ppicY}

\newcommand{\ParaWFig}[3]{%
   \settowidth{\ppicX}{#2}  
   \setlength{\ppicY}{\linewidth}
   \addtolength{\ppicY}{-\ppicX}
   \addtolength{\ppicY}{-0.30cm}
   \vspace{-0.2cm}%
   \SaveIndent%
   \noindent\makebox[\ppicY]{\begin{minipage}[t]{\ppicY}%
       \vspace{0pt}%
       \RestoreIndent%
       #1%
   \end{minipage}}%
   \hfill %
       \makebox[0.975\ppicX]{\begin{minipage}[t]{0.999\ppicX}%
       \vspace{0pt} %
       \hfill%
       #2

       #3
       \vfill
   \end{minipage}}
   \smallskip
}

\newcommand{\IntRange}[1]{\left\llbracket #1 \right\rrbracket}

\providecommand{\BibLatexMode}[1]{}
\providecommand{\BibTexMode}[1]{#1}

\ifx\UseBibLatex\undefined%
  \renewcommand{\BibLatexMode}[1]{}
  \renewcommand{\BibTexMode}[1]{#1}
\else
  \renewcommand{\BibLatexMode}[1]{#1}
  \renewcommand{\BibTexMode}[1]{}
\fi

\BibLatexMode{%
   \usepackage[bibencoding=ascii,style=alphabetic,backend=biber]{biblatex}%
   \usepackage{sariel_biblatex}%
}

\providecommand{\Mh}[1]{#1}%

\newcommand{\Tree}{\Mh{\mathsf{T}}}%
\newcommand{\Root}{\Mh{\mathsf{r}}}%
\newcommand{\Graph}{\Mh{G}}%
\newcommand{\tri}{\Mh{\eta}}%
\newcommand{\rdistX}[1]{\Mh{\ell}\pth{#1}}%

\newcommand{\VS}{\Mh{V}}%
\newcommand{\ES}{\Mh{E}}%
\newcommand{\FS}{\Mh{F}}%

\newcommand{\PathTY}[2]{\Mh{\pi}\pth{#1, #2}}%
\newcommand{\EX}[1]{\Mh{\ES}\pth{#1}}%
\newcommand{\VX}[1]{\Mh{\VS}\pth{#1}}%
\newcommand{\FX}[1]{\Mh{F}\pth{#1}}%
\newcommand{\Deep}{\Mh{\psi}}%
\newcommand{\PX}[1]{\Mh{P}_{\leq {#1}}}%

\newcommand{\Cycle}{\Mh{C}}
\newcommand{\CycleA}{\Mh{D}}%
\newcommand{\CycleB}{\Mh{J}}%
\newcommand{\CycleC}{\Mh{K}}%

\newcommand{\gCycle}{\Mh{{\xi}}}

\newcommand{\CS}{\Mh{S}}%

\newcommand{\fN}{\Mh{\mathsf{f}}}%

\newcommand{\inX}[1]{\Mh{\mathrm{in}}\pth{#1}}%
\newcommand{\outX}[1]{\Mh{\mathrm{out}}\pth{#1}}%

\newcommand{\ts}{\hspace{0.6pt}}
\newcommand{\concat}{\ts\Mh{\circ}\ts}%
\BibLatexMode{%
}%

\author{%
   Sariel Har-Peled%
   \SarielThanks{}%
   \and%
   Amir Nayyeri%
   \AmirThanks{}%
}%

\begin{document}

\title{A Simple Algorithm for Computing a Cycle Separator}%
\date{\today}%

\maketitle
\begin{abstract}
    We present a linear time algorithm for computing a cycle separator
    in a planar graph that is (arguably) simpler than previously known
    algorithms.  Our algorithm builds on, and is somewhat similar to,
    previous algorithms for computing separators. In particular, the
    algorithm described by Klein and Mozes \cite{km-oapg-17} is quite
    similar to ours.  The main new ingredient is a specific layered
    decomposition of the planar graph constructed differently from
    previous \BFS-based layerings.
\end{abstract}

\section{Introduction}
The planar separator theorem is a fundamental result in the study of
planar graphs that has been used in many divide and conquer
algorithms.  The theorem guarantees for planar graphs the existence of
$O(\sqrt{n})$ vertices whose removal breaks the graph into ``small''
pieces, connected components of size at most $\alpha n$ for a constant
$\alpha$.  For triangulated planar graphs, a stronger result is known
-- the separator is a simple cycle of length $O(\sqrt{n})$ whose
inside and outside (in the planar embedding) each contains at most
$\alpha n$ vertices.

The separator theorem was first proved by Ungar~\cite{u-tpg-51} with a
slightly weaker upper bound of $O(\sqrt{n}\log n)$.  Lipton and
Tarjan~\cite{lt-stpg-79} showed how to compute, in linear time, a
separator of size $O(\sqrt{n})$. Later, Miller~\cite{m-fsscs-86}
described a linear time algorithm for computing a cycle separator.

In this paper, we describe a simple algorithm for computing a cycle
separator.  We believe the simplicity of our algorithms is comparable
to that of the original algorithm of Lipton and Tarjan
\cite{lt-stpg-79}.

\paragraph{Existential proofs.}
 Alon \etal. \cite{ast-ps-94} described an
existential proof of the cycle separator theorem using a maximality
condition.

Miller \etal \cite{mttv-sspnng-97} showed how to compute a planar
separator in a planar graph if its circle packing realization is given
(this proof was later simplified by Har-Peled~\cite{h-speps-13}). In
particular, the planar separator theorem is an easy consequence of the
work of Paul Koebe \cite{k-kdka-36} (see \cite{h-speps-13} for
details).  A nice property of the proof of Miller \etal
\cite{mttv-sspnng-97}, is that it immediately implies the cycle
separator theorem. Unfortunately, there is no finite algorithm for
computing the circle packing realization of a planar graph -- all
known algorithms are iterative convergence algorithms. That is, the
proof of Miller \etal is an existential proof.

\paragraph{Constructive proofs.}
As mentioned above, Miller~\cite{m-fsscs-86} gave a linear time
algorithm for computing the cycle separator. A somewhat different
algorithm is also provided in the work of Klein \etal
\cite{kms-srsdp-13}, which computes the whole hierarchy of such
separators in linear time.  Fox-Epstein \etal \cite{fmps-sscsp-16}
also provides an algorithm for computing a cycle separator in linear
time.

\paragraph{This paper.}

A simple cycle is a \emph{$\alpha$-separator} if its inside and
outside each contains at most $\ceil{\alpha \fN}$ faces, where $\fN$
is the number of faces of the graph.  We present a linear time
algorithm for computing a cycle $2/3$-separator -- see \thmref{sep}.
The algorithm is somewhat similar in spirit to the work of Fox-Epstein
\etal \cite{fmps-sscsp-16}. A closer algorithm to ours is described by
Klein and Mozes \cite[Section 5.9]{km-oapg-17}. The new algorithm is
(arguably) slightly simpler than these previous versions.

The rest of the paper is composed of two section. In \secref{prelim}
we define some required basic concepts, and in \secref{proofs} we
describe the new algorithm.

\section{Preliminaries}
\seclab{prelim}

Let $\Graph$ be a triangulated planar graph embedded in the plane,
with vertex set $\VS$, edge set $\ES$, and face set $\FS$,
and let $\Graph^* = (\VS^*, \ES^*,\FS^*)$ be the
\emphi{dual} of $\Graph$.  A vertex $x\in \VS$ corresponds to a
face $x^*\in \FS^*$, an edge $xy \in \ES$ to an edge
$\pth{xy}^* \in \ES^*$, and a face $xyz\in \FS$ to a vertex
$(xyz)^* \in \VS^*$.  Because of the last correspondence, and
since $\Graph$ is triangulated, $\Graph^*$ is $3$-regular: all its
vertices have degree three.  For any spanning tree
$\Tree = (\ES_\Tree, \VS_\Tree)$ of $\Graph$, the duals of
the edges $\ES \setminus \ES_\Tree$ form a spanning tree of the
dual graph $\Graph^*$.

For any simple cycle $\Cycle$ in the embedding of $\Graph$, the
\emphi{inside} (resp., \emphi{outside}) of $\Cycle$, denoted by
$\inX{\Cycle}$ (resp., $\outX{\Cycle}$), is the bounded (resp.,
unbounded) region of $\Re^2 \setminus \Cycle$. Each vertex of $\VS$ is
inside, outside or on $\Cycle$.  A face is \emphi{inside}
(resp. \emphi{outside}) $\Cycle$ if its interior is a subset of
$\inX{\Cycle}$ (resp.  $\outX{\Cycle}$).  It follows that each face of
$\Graph$ is either inside or outside $\Cycle$.  If a face $\theta$ is
inside $\Cycle$, then $\Cycle$ \emphi{contains} $\theta$.

\begin{defn}
    For a cycle $\Cycle$, and an $1/2\leq\alpha<1$, $\Cycle$ is an
    \emphi{$\alpha$-cycle separator} of a graph $\Graph$, if the
    number of faces inside (resp. outside) $\Cycle$ is at most
    $\ceil{\alpha\cdot \fN}$, where $\fN = \cardin{\FS}$ is the number
    of faces of $\Graph$.
\end{defn}

For two cycles $\Cycle_1$ and $\Cycle_2$ of $\Graph$, $\Cycle_1$ is
\emphi{inside} $\Cycle_2$, denoted by $\Cycle_1\preceq \Cycle_2$, if
$\inX{\Cycle_1} \subseteq \inX{\Cycle_2}$.  For
$\Cycle_1\preceq \Cycle_2$, a face is \emphi{between} $\Cycle_1$ and
$\Cycle_2$, if it is inside $\Cycle_2$ and outside $\Cycle_1$.

Let $\gamma$ be a simple path or cycle in $\Graph$.  The
\emphi{length} of $\gamma$, denoted by $\lenX{\gamma}$, is the number
of edges of $\gamma$.  If $\gamma$ is a path, and $x,y$ are vertices
on $\gamma$, $\gamma[x,y]$ denotes the subpath of $\gamma$ between $x$
and $y$.  For two internally disjoint paths $\gamma_1$ and $\gamma_2$,
if the last vertex of $\gamma_1$ and the first vertex of $\gamma_2$
are identical, $\gamma_1\concat\gamma_2$ denotes the path or cycle
obtained by their \emphi{concatenation}


\section{The cycle separator theorem}
\seclab{proofs}

Let $\Graph=(\VS,\ES,\FS)$ be a triangulated planar graph embedded on
the plane, and let $n = |V|$, and $\fN = |\FS|$.  In this section, we
describe the linear time algorithm for computing a cycle separator of
$\Graph$.

Our construction is composed of three phases.  First, we find a
possibly long cycle separator $\CS$, by finding a spanning tree
$\Tree$ of $\Graph$, and a balanced edge separator $(uv)^*$ in its
dual tree.  The unique cycle in $\Tree \cup \{uv\}$ is guaranteed to
be a (possibly long) cycle separator (\secref{long:sep}).  This part
of the construction is similar to Lemma 2 of Lipton and Tarjan
\cite{lt-stpg-79}, and we include the details for completeness.  Next,
we build a nested sequence of cycles
$\Cycle_1\preceq \Cycle_2\preceq \ldots\preceq \Cycle_k$
(\secref{layers}).  The specific construction of these cycles, which
is guided by $\CS$, is the main new ingredient in the new algorithm.
Finally, we consider the collection of cycles
$\Cycle_1, \ldots, \Cycle_k$ and $\CS$, and construct a few short
cycles, such that one them is guaranteed to be a balanced separator
(\secref{construction}).

\subsection{A possibly long cycle separator}
\seclab{long:sep}

We start by computing a balanced separator that, unfortunately, can be
too long.  For a \BFS tree $\Tree$, we denote by $\PathTY{\Tree}{u}$
the unique shortest path in $\Tree$ between the root of $\Tree$ and
$u$.

\begin{lemma}[\cite{lt-stpg-79}]
    \lemlab{there:is:sack}%
    Given a triangulated planar graph $\Graph$, one can compute, in
    linear time, a \BFS tree $\Tree$ rooted at a vertex $\Root$, and
    an edge $uv \in \EX{\Graph}$, such that:
    \begin{compactenumA}
        \item the (shortest) paths $p_u =\PathTY{\Tree}{u}$ and
        $p_v =\PathTY{\Tree}{v}$ are edge disjoint,
        
        \item the cycle $S = p_u \cup p_v \cup uv$ is a
        $2/3$-separator for $\Graph$.
    \end{compactenumA}
\end{lemma}
\begin{proof}
    Our proof is a slight modification of the one provided by Lipton
    and Tarjan \cite{lt-stpg-79}, and we include it for the sake of
    completeness.  Let $r'\in \VS$ be any vertex, and let
    $\Tree = (\VS_\Tree, \ES_\Tree)$ be a \BFS tree rooted at $r'$.
    Also, let $D = \ES \setminus \ES_\Tree$, and note that the dual
    set of edges $D^*$ is a spanning tree of the dual $\Graph^*$.
    Since $\Graph$ is a triangulation, $D^*$ has maximum degree at
    most three.  Thus, it contains an edge $(uv)^*$ whose removal
    leaves two connected components, $D^*_{in}$ and $D^*_{out}$, each
    with at most $\ceil{(2/3)\fN}$ (dual) vertices, see
    \lemref{tree:sep}, where $\fN = \cardin{\FS}$ is the number of
    faces of $\Graph$. Let $D^*_{out}$ be the connected component that
    contains the dual of the outer face, and let $D^*_{in}$ be the
    other one.
    
    Let $uv$ be the original edge that is dual of $(uv)^*$, and $\CS$
    the unique cycle in $\Tree \cup \{uv\}$.  The sets of faces inside
    and outside $\CS$, correspond to the vertex sets of $D^*_{in}$ and
    $D^*_{out}$, respectively.  Thus, $\CS$ is a $2/3$-cycle
    separator.
    
    Now, let $\Root$ be the lowest common ancestor of $u$ and $v$ in
    $\Tree$.  The cycle $\CS$ is composed of $p_u = \Tree[r,u]$,
    $p_v = \Tree[r,v]$ and the edge $uv$.  Since $\Tree$ is a \BFS
    tree, and $\Root$ is an ancestor of $u$ and $v$, the paths $p_u$
    and $p_v$ are shortest paths in $\Graph$.
    
    To get a \BFS tree rooted at $\Root$, one simply recompute the
    \BFS tree starting from $\Root$, where we include the edges of
    $p_u$ and $p_v$ in the newly computed \BFS tree $\Tree$.
\end{proof}

For the rest of the algorithm, let $\CS$, $\Root$, $uv$, $p_u$ and
$p_v$ be as specified by \lemref{there:is:sack}.  We emphasize that
the graph is unweighted, $p_u$ and $p_v$ are shortest paths, and $u$
and $v$ are neighbors.


\begin{figure}[t]
    \centerline{%
       \begin{tabular}{cc}
         \includegraphics[page=1,width=0.45\linewidth]%
         {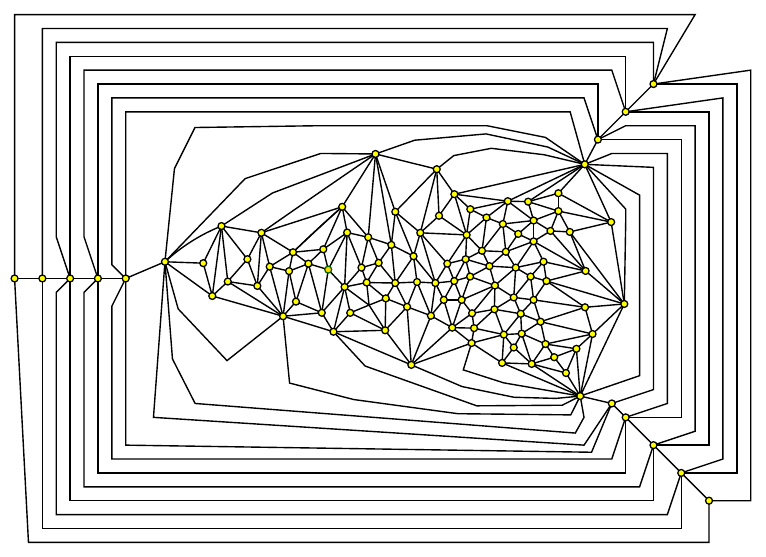}%
         \quad%
         &%
           \quad%
           \includegraphics[page=2,width=0.45\linewidth]%
           {figs/big_graph}%
       \end{tabular}%
    }
    \caption{A graph and its \BFS tree.}
\end{figure} 

\subsection{A nested sequence of short cycles}%
\seclab{layers}%

\begin{figure}[t]
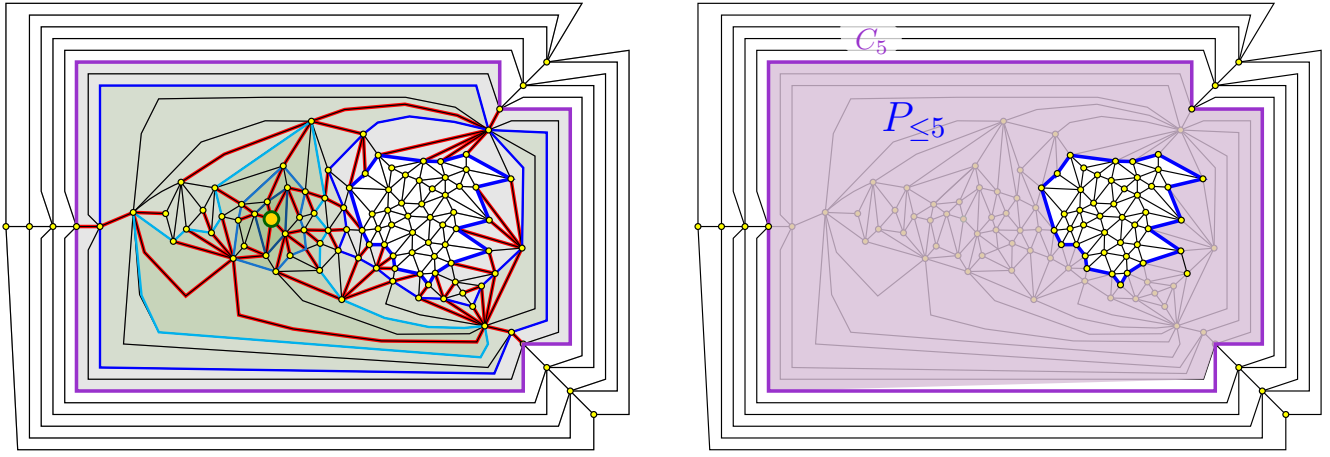

    \centerline{%
       \begin{tabular}{cc}
         \includegraphics[page=9,width=0.45\linewidth]%
         {figs/big_graph}%
         \quad%
         &%
           \quad%
           \includegraphics[page=10,width=0.45\linewidth]%
           {figs/big_graph}%
       \end{tabular}%
    }
    \caption{The region $\PX{5}$ and the associated outer cycle
       $\Cycle_5$.}
\end{figure}

\newcommand{\hT}{\Mh{\mathsf{h}}}%

Let $\Root$ be the root node of the \BFS tree $\Tree$ computed by
\lemref{there:is:sack}.  For $x \in \VX{\Graph}$, let $\rdistX{x}$ be
the distance in $\Tree$ of $x$ from the root $\Root$. The
\emphi{level} of a (triangular) face $\tri = xyz$ of $\Graph$ is
$\rdistX{\tri} = \max\pth{ \rdistX{x}, \rdistX{y}, \rdistX{z}}$. In
particular, a face $\tri = uvz \in \FX{\Graph}$ is \emphi{$i$-close}
to $\Root$ if $\rdistX{\tri} \leq i$. The union of all $i$-close
faces, form a region $\PX{i}$ in the plane\footnote{Here,
   conceptually, we consider the embedding of the edges of $\Graph$ to
   be explicitly known, so that $\PX{i}$ is well defined. The
   algorithm does not need this explicit description.}. This region is
simple, but it is not necessarily simply connected.

Let $\hT = \max \pth{ \rdistX{u}, \rdistX{v}}$, and let
$\Deep \in \brc{u,v}$ be the vertex realizing $\hT$.  We assume, for
the sake of simplicity of exposition, that $\Deep$ is one of the
vertices of the outer face\footnote{This can be ensured by applying
   inversion to the given embedding of $\Graph$ -- but it is not
   necessary for our algorithm.}.

For $i < \hT$, let $\gCycle_i$ be the outer connected component of
$ \partial \PX{i}$. This is a closed curve in the plane, with $\Deep$
being outside it (as long as $i < \hT$), and let $\Cycle_i$ be the
corresponding cycle of edges in $\Graph$ that corresponds to
$\gCycle_i$. The resulting set of cycles is
$\Cycle_0, \ldots, \Cycle_{\hT-1}$ (i.e., a cycle $\Cycle_i$ is empty
if $i \geq \hT$).

\begin{lemma}
    \lemlab{disjoint:B:s}%
    We have the following:
    \begin{compactenumA}
        \item For any $i < \hT$, the vertices of $\Cycle_i$ are all at
        distance $i$ from $\Root$ in $\Tree$.
        
        \item For any $i < \hT$, the cycle $\Cycle_i$ is simple.
        
        \item For any $i < j < \hT$, the cycles $\Cycle_i$ and
        $\Cycle_j$ are vertex disjoint.

        \item \itemlab{D}%
        For $i < \hT$, the cycle $\Cycle_i$ intersects the cycle
        $\CS$.
    \end{compactenumA}
\end{lemma}
\begin{proof}
    (A) Consider a vertex $x$ in $\Graph$ with $\rdistX{x} <i$. As
    $\Tree$ is a \BFS tree, we have that all the neighbors $y$ of $x$
    in $\Graph$, have $\rdistX{y} \leq \rdistX{x} + 1 \leq i$. Namely,
    all the triangles adjacent to $x$ are $i$-close, and the vertex
    $x$ is internal to the region $\PX{i}$, which implies that it can
    not appear in $\Cycle_i$.
    
    \ParaWFig{%
       (B) Since $\gCycle_i$ is the (closure) of the outer boundary of
       a connected set, the corresponding cycle of edges $\Cycle_i$ is
       a cycle in the graph. The bad case here is that a vertex $x$ is
       repeated in $\Cycle_i$ more than once.  But then, $x$ is a cut
       vertex for $\PX{i}$ -- removing it disconnects $\PX{i}$ -- see
       \figref{butterfly}. Now, $\rdistX{x} < i$ as the \BFS from
       $\Root$ must have passed through $x$ from one side of $\PX{i}$
       to the other side. Arguing as in (A), implies that $x$ is
       internal to $\PX{i}$, which is a contradiction.%
    }%
    {~\includegraphics{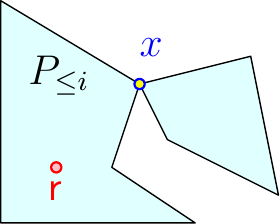}}%
    {\captionof{figure}{} %
       \figlab{butterfly}}

    (C) is readily implied by (A).

    \medskip%
    (D) Indeed, $\Cycle_i$ must intersect the shortest path $p_\Deep$
    from $\Root$ to $\Deep$, and as this path is part of $\CS$, the
    claim follows.
\end{proof}

Computing the cycles $\Cycle_i$, for all $i$, can be done in linear
time (without the explicit embedding of the edges of $\Graph$). To
this end, compute for all the (triangular) faces of $\Graph$ their
level. Next, mark all the edges between faces of level $i$ and $i+1$
as boundary edges forming $\partial \PX{i}$ -- this yields a
collection of cycles. To identify the right cycle, consider the
shortest $p_{\Deep}$ path between $\Root$ and $\Deep$. The cycle with
a vertex that belongs to $p_{\Deep}$ is the desired cycle $\Cycle_i$.
Clearly, this can be done in linear time overall for all these cycles.

\begin{lemma}
    \lemlab{nested:seq}%
    Let $\Delta > 0$ be an arbitrary parameter.  If
    $\hT = \rdistX{\Deep} > \Delta$, then there exist an integer
    $i_0 \in \IntRange{\Delta}$, such that
    $\cardin{\Cycle^{}_{i_0}} > 0$ and
    $\sum_{j\geq 0} \cardin{\Cycle^{}_{i_0 + j\Delta}} \leq n/\Delta$,
    where $\cardin{\Cycle^{}_{k}}$ denotes the number of vertices of
    $\Cycle_k$.
\end{lemma}
\begin{proof}
    Setting $g(i) = \sum_{j\geq 0} \cardin{\Cycle^{}_{i + j\Delta}}$.
    By \lemref{disjoint:B:s} \itemref{D}, $g(i) > 0$, for
    $i =0,\ldots, \Delta-1$.  We have
    \begin{align*}
      \sum_{i=0}^{\Delta-1} g(i) %
      \leq%
      \sum_{i=0}^{\Delta-1}\sum_{j\geq 0} \cardin{\Cycle^{}_{i +
      j\Delta}} =%
      \sum_{k\geq0}^{\hT-1} \cardin{\Cycle^{}_{k}} \leq%
      |\VX{\Graph}| %
      \leq n,
    \end{align*}
    as the cycles $\Cycle_0, \Cycle_1, \ldots, \Cycle_{\hT-1}$ are
    disjoint.  As such, there must be an index $i = i_0$ of the first
    summation that does not exceed the average.
\end{proof}


\subsection{Constructing the cycle separator}
\seclab{construction}

\subsubsection{The algorithm}

Let $\Delta =\Theta(\sqrt{n})$ be a parameter to be specified shortly.
Let $\CS$ be a $2/3$-cycle separator, and $\Root$, $u$, $v$, $p_u$,
and $p_v$ as specified by \lemref{there:is:sack}.  If
$\cardin{\CS} \leq 2\Delta$ then this is the desired a short cycle
separator. So, assume that
\begin{math}
    \hT \geq |{\CS}|/2 > \Delta.
\end{math}

For $j\geq 0$, let $\alpha_j = i_0 + (j-1)\Delta$ be the index of the
$j$\th cycle in the small ``ladder'' of \lemref{nested:seq}.  Since
$\hT > \Delta$ and by \lemref{disjoint:B:s} \itemref{D}, the cycles
$\Cycle _{i_0} = \Cycle_{\alpha_0}$ of the ladder intersects $\CS$.
In particular, let $\CycleA_j = \Cycle_{\alpha_j}$, for
$j=1, \ldots, k-1$, be the $j$\th nested cycles of this light ladder
that intersects $\CS$. Specifically, let $k$ the minimum value such
that $\alpha_k \geq \hT$.  Let $\CycleA_0$ be the trivial cycle formed
by the root vertex $\Root$. Similarly, let $\CycleA_k$ be the trivial
cycle formed only by the vertex $\Deep$, such that its interior
contains the whole graph.

For $j=0,\ldots, k$, let $\fN_j$ be the number of faces in the
interior of $\CycleA_j$. If for some $j$, we have that
$\floor{\fN/3} \leq \fN_j \leq \ceil{(2/3)\fN}$, then $\CycleA_j$ is
the desired separator, as its length is at most $n/\Delta$ by
\lemref{disjoint:B:s}, where $\fN$ is the number of faces of $\Graph$.

Otherwise, there must be an index $i$, such that $\fN_i < \fN/3$, and
$\fN_{i+1} > (2/3)\fN$. Assume, for the sake of simplicity of
exposition that $0 < i < k-1$ (the cases that $i=0$ or $i=k-1$ are
degenerate and can be handled in a similar fashion to what follows).

\begin{figure}[h]
    \centerline{%
       \includegraphics[page=1]{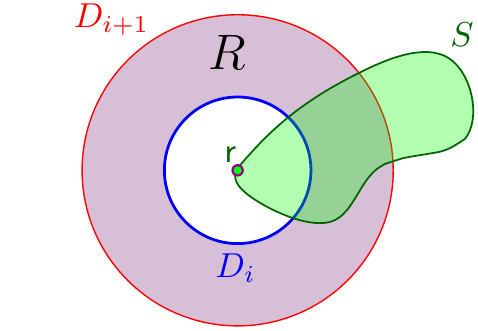}\qquad
       \includegraphics[page=2]{figs/heavy_ring}%
    }
    \caption{}
    \figlab{ring}%
\end{figure}

Consider the ``heavy'' ring $R$ bounded by the two of the nested
cycles $\CycleA_{i+1}$ and $\CycleA_{i}$, see \figref{ring}.


\begin{observation}
    By \lemref{disjoint:B:s}, the cycles $\CycleA_i$ and
    $\CycleA_{i+1}$ each intersects $\CS$ in two vertices exactly. And
    $\CycleA_i$ is nested inside $\CycleA_{i+1}$.
\end{observation}

\medskip%
\ParaWFig{%
   
   Let $I_i$ and $O_i$ the portions of $\CycleA_i$ inside and outside
   $\CS$, respectively (define $I_{i+1}$ and $O_{i+1}$ similarly).
   Let $p_i$ and $q_i$ (resp., $p_{i+1}$ and $q_{i+1}$) be the
   endpoints of $I_i$ (resp., $I_{i+1}$), such that $p_i$ is adjacent
   to $p_{i+1}$ along $\CS$.
   
   We can now partition $R$ into two cycles $R_1$ and $R_2$. The
   region $R_1$ is bounded by the cycle formed by
   $\CycleB_1 = \CS[q_i, q_{i+1}] \concat I_{i+1} \concat \CS[p_{i+1},
   p_i] \concat I_i$.  The region $R_2$ is bounded by the cycle formed
   by
   \begin{math}
       \CycleB_2 = \CS[q_i, q_{i+1}] \concat O_{i+1} \concat
       \CS[p_{i+1}, p_i] \concat O_i,
   \end{math}
   see \figref{R:1:R:2}. %
   
   We have that
   \begin{math}
       \cardin{\CycleB_1}%
       \leq%
       \cardin{\CycleA_i} + \cardin{\CycleA_{i+1}} + 2\Delta%
       \leq%
       n/\Delta + 2\Delta,
   \end{math}
   by \lemref{nested:seq}.  In particular, if $\fN(R_1) \geq \fN/3$,
   then $\CycleB_1$ is the desired cycle separator, since
   $\fN(R_1) \leq \fN(\CS) \leq \ceil{(2/3)\fN}$.
   
}%
{\includegraphics[page=3]{figs/heavy_ring}}%
{\captionof{figure}{}%
   \figlab{R:1:R:2}%
}

Similarly, if $\fN(R_2) \geq \fN/3$, then $\CycleB_2$ is the desired
cycle separator, since
\begin{math}
    \fN(R_2)%
    \leq%
    \fN - \fN(\CS)%
    \leq%
    \ceil{(2/3)\fN}.
\end{math}

Otherwise, the algorithm returns the cycle $\CycleC$ formed by
\begin{math}
    O_i \concat%
    \CS[q_i, q_{i+1}]%
    \concat%
    I_{i+1}
    \concat%
    \CS[p_{i+1}, p_i]
\end{math}
as the desired separator.

\begin{figure}[h]
    \centerline{%
       \begin{tabular}{cc}
         \includegraphics[page=5]{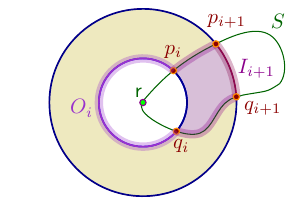}%
         &
           \includegraphics[page=4]{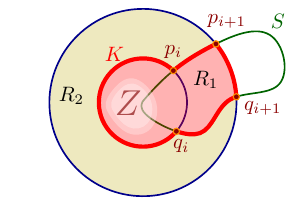}%
       \end{tabular}
    }
    \caption{}
    \figlab{ring:2}%
\end{figure}

\subsubsection{Analysis}

\begin{lemma}
    \lemlab{w:case}%
    Assume that $\fN(R_1) < \fN/3$ and $\fN(R_2) < \fN/3$. Consider
    the region $Z$, formed by the union of the interior of
    $\CycleA_i$, together with the interior of $R_1$. Its boundary, is
    the cycle $\CycleC$ formed by
    \begin{math}
        O_i \concat%
        \CS[q_i, q_{i+1}]%
        \concat%
        I_{i+1}
        \concat%
        \CS[p_{i+1}, p_i],
    \end{math}
    see \figref{ring:2}. The cycle $\CycleC$ is a $2/3$-cycle
    separator with $n/\Delta + 2\Delta$ edges.
\end{lemma}
\begin{proof}
    We have the following: \smallskip%
    \begin{enumerate*}[label=(\roman*)]
        \item $\fN_i < \fN/3$,
        \item $\fN_i + \fN(R_1) + \fN(R_2) = \fN_{i+1} > (2/3)\fN$,
        \item $\fN(R_1) < \fN/3$, and
        \item $\fN(R_2) < \fN/3$.
    \end{enumerate*}
    \smallskip%
    Assume that $\fN_i + \fN(R_1) < \fN/3$. But then
    $\fN_{i+1} = \fN_i + \fN(R_1) + \fN(R_2) < (2/3)\fN$, which is
    impossible. The region $Z$ bounded by $\CycleC$ contains
    $\fN_i + \fN(R_1)$ faces, and we have
    $\fN/3 < \fN_i + \fN(R_1) < (2/3)\fN$, which implies the separator
    property.
    
    As for the length of $\CycleC$, observe that 
    \begin{math}
        \cardin{\CycleC}%
        \leq%
        \cardin{\CycleA_i} + \cardin{\CycleA_{i+1}} + \cardin{
           \CS[p_i, p_{i+1}]} + \cardin{ \CS[q_i, q_{i+1}]}%
        \leq %
        n/\Delta + 2\Delta,
    \end{math}
    by \lemref{nested:seq}.
\end{proof}

\begin{theorem}
    \thmlab{sep}%
    Given an embedded triangulated planar graph $\Graph$ with $n$
    vertices and $\fN$ faces, one can compute, in linear time, a
    simple cycle $\CycleC$ that is a $2/3$-separator of $\Graph$. The
    cycle $\CycleC$ has at most $O(1) + \sqrt{8n}$ edges.
    
    This cycle $\CycleC$ also $2/3$-separates the vertices of $\Graph$
    -- namely, there are at most ${(2/3)n}$ vertices of $\Graph$ on
    each side of it.
\end{theorem}
\begin{proof}
    The construction is described above. As for the length of
    $\CycleC$, set $\Delta = \ceil{\bigl.\smash{\sqrt{n/2}}\,}$.  By
    \lemref{w:case}, we have
    \begin{math}
        \cardin{\CycleC}%
        \leq%
        2\Delta + n /\Delta%
        \leq%
        O(1) + \sqrt{2n} + {\sqrt{2n}}%
        \leq%
        O(1) + {\sqrt{8n}}.
    \end{math}
    (The separator cycle is even shorter if one of the other cases
    described above happens.)
    
    As for the running time, observe that the algorithm runs \BFS on
    the graph several times, identify the edges that form the relevant
    cycles. Count the number of faces inside these cycles, and finally
    counts the number of edges in $R_1$ and $R_2$. Clearly, all this
    work (with a careful implementation) can be done in linear time.
    
    The second claim follows from a standard argument, see
    \lemref{good:v:sep} \itemref{part:c} below for details.
\end{proof}

\subsection{From faces separation to vertices separation}

\begin{lemma}
    \lemlab{good:v:sep}%
    %
    \begin{enumerate}[label=(\Alph*),parsep=-0.5ex]
        \item A simple planar graph $\Graph$ with $n$ vertices has at
        most $3n-6$ edges and at most $2n - 4$ faces. A triangulation
        has exactly $3n-6$ edges and $2n-4$ faces.
        
        \item Let $\Graph$ be a triangulated planar graph and let
        $\Cycle$ be a simple cycle in it.  Then, there are exactly
        $ (\fN(\Cycle)-\cardin{\Cycle})/2 +1$ vertices in the interior
        of $\Cycle$, where $\fN(\Cycle)$ denotes the number of faces
        of $\Graph$ in the interior of $\Cycle$.
        
        \item \itemlab{part:c} A simple cycle $\Cycle$ in a
        triangulated graph $\Graph$ that has at most $\ceil{(2/3)\fN}$
        faces in its interior, contains at most $(2/3)n$ vertices in
        its interior, where $n$ and $\fN$ are the number of vertices
        and faces of $\Graph$, respectively.
    \end{enumerate}
\end{lemma}
\begin{proof}
    (A) is an immediate consequence of Euler's formula.

    (B) Let $n'$ be the number of vertices of $\Graph$ in or on
    $\Cycle$ -- delete the portion of $\Graph$ outside $\Cycle$, and
    add a vertex $v$ to $\Graph$ outside $\Cycle$, and connect it to
    all the vertices of $\Cycle$.  The resulting graph is a
    triangulation with $n'+1$ vertices, and $2(n'+1)-4=2n'-2$
    triangles, by part (A).  This counts $\cardin{\Cycle}$ triangles
    that were created by the addition of $v$. As such,
    $\fN(\Cycle) = 2n'-2 - \cardin{\Cycle}$ $\implies$
    $n' = \fN(\Cycle)/2+1 + \cardin{\Cycle}/2$. The number of vertices
    inside $C$ is
    $n' - \cardin{\Cycle} = (\fN(\Cycle)-\cardin{\Cycle})/2 +1$.
    
    (C) Part (B) implies that number of vertices inside the region
    formed by the cycle $\Cycle$ is 
    \begin{align*}
      (\fN(\Cycle)-\cardin{\Cycle})/2 +1
      &  \leq%
        (\ceil{(2/3)\fN}-\cardin{\Cycle})/2 +1
        =%
        (\ceil{(2/3)(2n-4)}-\cardin{\Cycle})/2 +1
      \\
      &\leq%
        \frac{(2/3)(2n-4) + 1-\cardin{\Cycle}}{2} +1
        \leq%
        \frac{2}{3} n,
    \end{align*}
    as claimed.
\end{proof}

\BibTexMode{%
 \providecommand{\CNFX}[1]{ {\em{\textrm{(#1)}}}}

}

\BibLatexMode{\printbibliography}

\appendix%

\section{Balanced edge separator in %
   a low-degree tree}
\apndlab{tree:sep:proof}

The following lemma is well known, and we provide a proof for the sake
of completeness.
\begin{lemma}
    \lemlab{tree:sep}%
    Let $\Tree$ be a tree with $n$ vertices, with maximum degree
    $d \geq 2$. Then, there exists an edge whose removal break $\Tree$
    into two trees, each with at most $\ceil{ (1-1/d)n}$
    vertices. This edge can be computed in linear time.
\end{lemma}
\begin{proof}
    Let $v_1$ be an arbitrary vertex of $\Tree$, and root $\Tree$ at
    $v_1$. For a vertex $v$ of $\Tree$ let $n(v)$ denote the number of
    nodes in its subtree -- this quantity can be precomputed, in
    linear time, for all the vertices in the tree using \DFS.
    
    In the $i$\th step, $v_{i+1}$ be the child of $v_i$ with maximum
    number of vertices in its subtree. If
    $n(v_{i+1})\leq \ceil{ (1-1/d)n}$, then the algorithm outputs the
    edge $x y$ as the desired edge separator, where $x= v_{i}$ and
    $y=v_{i+1}$. Otherwise, the algorithm continues the walk down to
    $v_{i+1}$.  Since the tree is finite, the algorithm stops and
    output an edge.
    
    Assume, for the sake of contradiction, that $n(y) < {n/d}$. But
    then, $x$ has at most $d(x)-1\leq d-1$ children (in the rooted
    tree), each one of them has at most $n(y)$ nodes (since $y$ was
    the ``heaviest'' child). As such, we have
    $n(x) \leq 1 + (d-1)n(y) < 1 + (d-1){n/d} \leq \ceil{ (1-1/d)n}$
    if $d$ does not divides $n$. If $d$ divides $n$ then
    \begin{math}
        n(x)%
        \leq%
        1 + (d-1)n(y)%
        \leq%
        1 + (d-1)(n/d - 1)%
        = ((d-1)/d)n +2 - d \leq%
        \ceil{ (1-1/d)n}.
    \end{math}

    Namely, the algorithm would have stopped at $x$, and not continue
    to $y$, a contradiction.
    
    As such, $n/d \leq n(y) \leq \ceil{(1-1/d) n}$. But this implies
    that $xy$ is the desired edge separator.
\end{proof}

\remove{%
   \newpage
   \begin{align*}
     f(x) =2x + n/x\\
     f'(x) =2 - nx^{-2}\\  
     2x^2= n\\  
     x= \sqrt{n/2}\\  
     f(x) = 2 + 2\sqrt{2n} + 2\sqrt{n}
   \end{align*}
}

\end{document}